\documentclass[twocolumn,pre,showpacs,nofootinbib,superscriptaddress]{revtex4}

\usepackage{graphicx,color}
\usepackage{amsmath,amssymb,latexsym,amsthm}
\usepackage{mathtools}
\usepackage{tikz}
\usetikzlibrary{chains,shapes,fit,calc,arrows}
\usepackage{pgfplots, pgfplotstable}
\usepgfplotslibrary{groupplots} 

\newtheorem{theorem}{Theorem}
\newtheorem{proposition}{Proposition}

\newtheorem{corollary}{Corollary}
\theoremstyle{definition}
\newtheorem{definition}{Definition}
\theoremstyle{remark}
\newtheorem{remark}{Remark}

\usepackage{color}

\definecolor{webgreen}{rgb}{0,.5,0}
\definecolor{webbrown}{rgb}{.6,0,0}
\definecolor{grigio}{rgb}{.85,.85,.85} 
\definecolor{RoyalBlue}{rgb}{0.0, 0.14, 0.4}
\definecolor{skyblue1}{rgb}{0.45,0.62,0.81}
\definecolor{skyblue2}{rgb}{0.2,0.39,0.64}
\definecolor{skyblue3}{rgb}{0.13,0.29,0.53}
\definecolor{scarlet1}{rgb}{0.93,0.16,0.16}
\definecolor{scarlet2}{rgb}{0.8,0,0}
\definecolor{scarlet3}{rgb}{0.64,0,0}

\begin{document}
\title{
Analog and Symbolic Computation through the Koopman Framework\\
}

\author{Francesco Caravelli}
\affiliation{University of Pisa, Largo Bruno Pontecorvo 3, 56127 Pisa, Italy}
\affiliation{NEST, Scuola Normale Superiore, Piazza San Silvestro 12, 56127 Pisa, Italy}

\author{Jean-Charles Delvenne}
\affiliation{Institute of Information and Communication Technologies, Electronics and Applied Mathematics, Université catholique de Louvain, Louvain-la-Neuve, Belgium}


\begin{abstract}
We develop a Koopman operator framework for studying the {computational properties} of dynamical systems. 
Specifically, we show that the resolvent of the Koopman operator provides a natural abstraction of halting, yielding a ``Koopman halting problem that is recursively enumerable in general. For symbolic systems, such as those defined on Cantor space, this operator formulation captures the reachability between clopen sets, while for equicontinuous systems we prove that the Koopman halting problem is decidable. Our framework demonstrates that absorbing (halting) states {in finite automata} correspond to Koopman eigenfunctions with eigenvalue one, while cycles in the transition graph impose algebraic constraints on spectral properties. These results provide a unifying perspective on computation in symbolic and analog systems, showing how computational universality is reflected in operator spectra, invariant subspaces, and algebraic structures. Beyond symbolic dynamics, this operator-theoretic lens opens pathways to analyze {computational power of} a broader class of dynamical systems, including polynomial and analog models, and suggests that computational hardness may admit dynamical signatures in terms of Koopman spectral structure.
\end{abstract}


\maketitle


\section{Introduction}

The foundations of computation theory were established in the 1930s with the equivalence of several intuitive notions of an ``algorithm'': Church’s $\lambda$-calculus, Kleene’s partial recursive functions, and Turing machines \cite{church1936unsolvable,turing1936computable}. Among these, Turing machines are especially notable for their dynamical systems character: they can be viewed as physical processes, where algorithmic execution reduces to questions of reachability or halting in a discrete state evolution \cite{moore1990unpredictability}. Other formal models, such as finite automata or pushdown automata, also define discrete trajectories, though with weaker computational power than Turing machines \cite{hopcroft1979automata}. Already in the 1940s, formal attempts at defining analog computing systems appeared \cite{shannon1941,bournez2007survey}.  

This naturally raises a broader question: which dynamical systems can genuinely be said to perform computation?  

Extensive work has explored this question in a variety of settings: differential equations \cite{rubel1981universal}, discrete-time piecewise linear maps \cite{branicky1995universal}, cellular automata \cite{wolfram1984universality}, and symbolic dynamical systems \cite{delvenne2006computational}. Yet no consensus exists on how to define the notion of ``halting'' in such systems. Should it mean state-to-state reachability? Eventual convergence to an absorbing set? And how should one address the practical impossibility of infinite precision in uncountable state spaces \cite{blum1989theory}? Depending on the interpretation, one obtains quite different answers to what it means for a dynamical system to ``compute.''  

Rather than resolving these debates, we suggest that the Koopman operator provides a unifying, high-level lens \cite{mezic2005spectral,budisic2012applied}. By lifting nonlinear dynamics into an infinite-dimensional linear operator framework acting on observables, the Koopman formalism allows halting and decidability to be recast in terms of spectral properties, invariant subspaces, and reachability of observables.  

This viewpoint is particularly appealing in light of recent results linking computational hardness in SAT problems with chaotic dynamics in analog systems. In particular, Ercsey-Ravasz and Toroczkai showed that the SAT--UNSAT transition in an analog solver\footnote{The mapping between the 3SAT and the analog system was done via an arithmetization procedure (analog AND(a,b)=ab, analog OR(a,b)=a+b-ab and analog NOT(a)=1-a) and then converting the clauses into a dynamical system, while forcing the system to converge only to values of 0 and 1s.} for 3SAT coincides with the onset of transient chaos in specially designed analog dynamical systems \cite{ercsey2011optimization}. The complexity of analog computation has been studied from multiple perspectives over the years \cite{pourel1974,vergis1986,siegelmann1995,moore1996,graca2003}, and remains an active line of research today \cite{bournez2006,bournez2017}.  

These dynamical regimes exhibit hallmarks of chaos such as sensitivity to initial conditions, fractal basin boundaries, and exponential divergence of trajectories. In this light, computational complexity appears as a kind of dynamical phase transition, with complexity-theoretic properties leaving physical signatures in the dynamics \cite{krzakala2007gibbs}.  

Crucially, while all NP problems are decidable in the classical Turing sense \cite{papadimitriou1994complexity}, their analog realizations may display chaotic dynamics that mirror the intrinsic difficulty of solving them. This raises the intriguing possibility that chaos itself could encode hardness, and that the Koopman operator (via its spectral structure) might provide a systematic tool to classify and quantify such computational phenomena across symbolic and analog systems.  

Our aim is thus to provide a systematic operator-theoretic perspective on the interplay between dynamical systems and computation, bridging symbolic and analog paradigms. 

In this article, we advance some arguments towards this direction. Specifically we model computing dynamical systems in at least two ways (autonomous dynamical systems such as Turing machines ; dynamical systems with external inputs such as finite automata), and in each case we find natural formalizations of standard computability properties such as the halting properties into properties of the Koopman operator. 

The structure of the present manuscript is as follows. We begin by recalling the definition of the Koopman operator for discrete-time autonomous systems. 
We propose a generalized halting problem (as it applies to Turing machines, cellular automata or other autonomous dynamical systems) expressed through the resolvent of the Koopman operator, and prove that it is recursively enumerable in the sense of Pour-El and Richards \cite{pourel1981wave,pourel1989computability}, while for equicontinuous symbolic systems it becomes decidable \cite{delvenne2006computational}. 

We then show how classical computing models with inputs (e.g., finite automata) can also be naturally embedded in this framework. In particular, we show how standard computational features manifest in Koopman terms, e.g.  halting states appear as eigenfunctions with eigenvalue one.

\begin{figure*}
    \centering
    \includegraphics[width=0.9\linewidth]{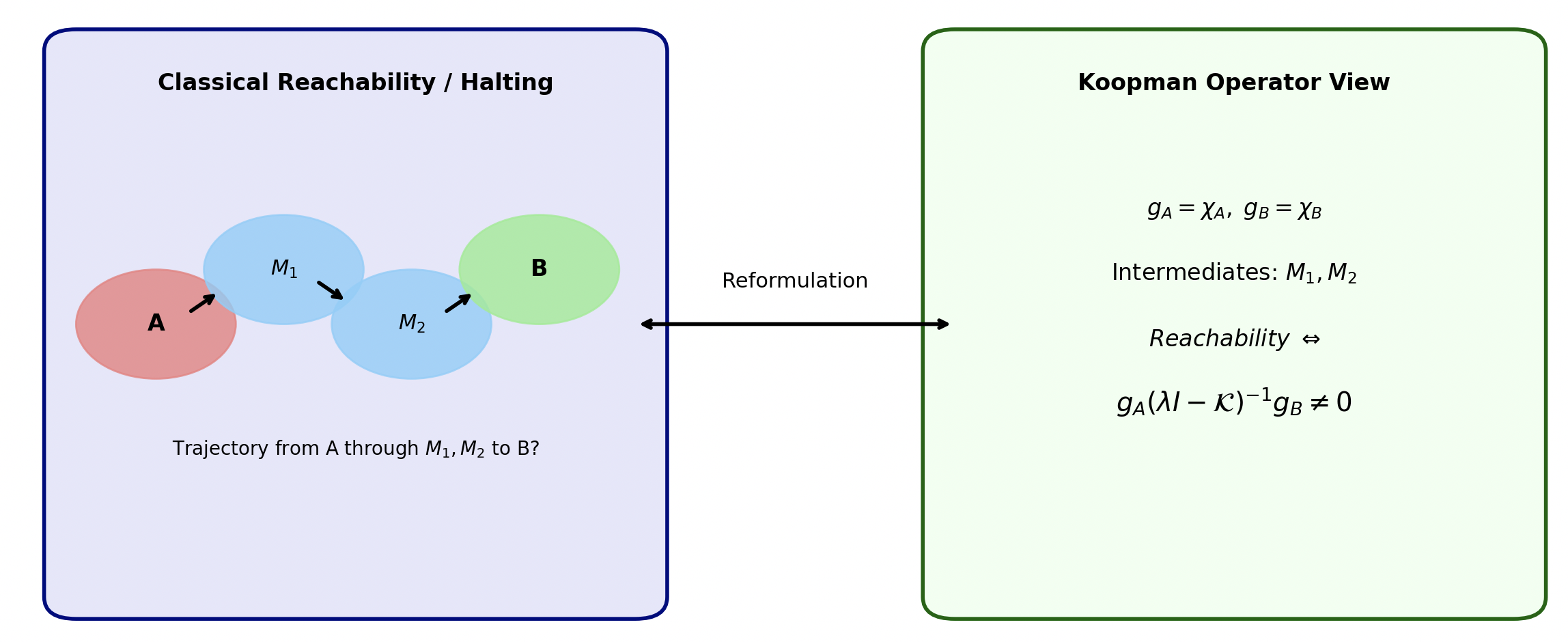}
    \caption{Set-to-set reachability (left) with intermediate sets $M_1, M_2$ is reformulated in the Koopman framework (right) as the resolvent condition $g_A(\lambda I - \mathcal{K})^{-1} g_B \neq 0$, connecting halting/reachability to operator theory.}
    \label{fig:1}
\end{figure*}

\section{Computation from the Koopman point of view}
\subsection{Koopman operator for dynamical systems }
\label{sec:koopman-intro}

In this section, we provide the definitions of the Koopman operator and state some of its basic properties of interest.

Consider a discrete‐time autonomous  dynamical system of the form
\begin{equation}
  \mathbf{x}_{t+1} \;=\; \mathbf{F}\bigl(\mathbf{x}_t\bigr),
  \quad
  \mathbf{x}_t \in X
  \quad
  t\in \mathbb{N}.
  \label{eq:discrete-map}
\end{equation}
We assume $\mathbf{F}:X \to X$ is a (possibly nonlinear) map defining the time evolution from step $t$ to $t+1$. In many applications $X$ may be a part of $\mathbb{R}^n$ or, as we shall see and detail later, the Cantor space.

Instead of tracking trajectories $\mathbf{x}(t)$ in phase space, \emph{Koopman theory} focuses on the evolution of \emph{observables}, i.e.\ scalar‐valued functions $g:X\to\mathbb{C}$. 
%
The \emph{Koopman operator} $\mathcal{K}$ associated with \eqref{eq:discrete-map} acts on observables $g$ via
\begin{equation}
  (\mathcal{K} g)\bigl(\mathbf{x}\bigr)
  \;=\;
  g\bigl(F(\mathbf{x})\bigr).
  \label{eq:koopman-definition}
\end{equation}
Because composition is linear in $g$, the map $g\mapsto \mathcal{K} g$ is a \emph{linear} operator, even though $\mathbf{F}$ may be highly nonlinear.  

Thus $\mathcal{K}$ is an infinite‐dimensional linear operator acting on a suitable function space (e.g.\ $L^2(X)$ or $C(X)$)


A principal goal in Koopman analysis is to find \emph{eigenfunctions} $\varphi(\mathbf{x})$ of the operator $\mathcal{K}$, satisfying

\begin{equation}
  (\mathcal{K}\varphi)\bigl(\mathbf{x}\bigr)
  \;=\;
  \varphi\bigl(\Phi(\mathbf{x})\bigr)
  \;=\;
 \lambda\,\varphi(\mathbf{x}),
  \label{eq:koopman-eigenfunction}
\end{equation}
%
where $\lambda\in\mathbb{C}$ is the associated eigenvalue.  These $\varphi(\mathbf{x})$ can be viewed as \emph{coordinates} in which the dynamics becomes effectively ``linear'' via multiplication by $\lambda$.  If such eigenfunctions can be identified (analytically or numerically), one obtains a spectral decomposition that can greatly simplify the study of long‐term dynamics.

This operator viewpoint has several advantages. First,
 despite $\mathbf{F}$ being nonlinear, $\mathcal{K}$ is linear in $g$.
Second, the operator’s spectrum (eigenvalues/eigenfunctions) encodes fundamental modes of the underlying dynamics, including slow/fast scales and resonances.





In many cases the dynamics unfold in continuous-time.
Consider a continuous-time \emph{autonomous} dynamical system on a state space $X\subseteq\mathbb{R}^n$:
\begin{equation}
  \frac{d\mathbf{x}}{dt} \;=\; \mathbf{f}\bigl(\mathbf{x}\bigr),
  \quad
  \mathbf{x}\in X.
  \label{eq:autonomous-system}
\end{equation}
Here, $\mathbf{f}:X\to\mathbb{R}^n$ is a (possibly nonlinear) vector field that governs the time evolution of the state $\mathbf{x}(t)$. In these case a Koopman theory also exists. In this article we focus on the discrete-time case.


The Koopman viewpoint has been introduced by Koopman and von Neumann in order to classify (discrete-time) systems, taking the eigenvalues as invariants \cite{koopman1931hamiltonian,koopman1932dynamical}. 
More recently, the Koopman viewpoint has been applied for data-driven analysis methods in dynamical systems and control theory \cite{mauroy2020koopman}, and also in the combinatorial studies of well-known dynamical systems \cite{caravellilin}. 




 \subsection{Computers as dynamical systems}


We now examine how the classical \emph{halting problem} can be reformulated within the Koopman operator perspective in a natural and rigorous way.  

Consider a discrete-time autonomous dynamical system  
\begin{equation}
\mathbf{F}: X \to X,
\end{equation}  
where $X$ is a compact metric space and $\mathbf{F}$ is continuous. In this setting, the state space $X$ represents the collection of all possible configurations of a computer, while the map $\mathbf{F}$ encodes the elementary step of computation, i.e., the update rule that advances the machine from one configuration to the next. Thus, each iteration of $\mathbf{F}$ corresponds to a unit of computational time, and the entire execution of the algorithm can be viewed as a trajectory in state space.  

The canonical example of such a computational model is the Turing machine, whose configuration is specified by the contents of its tape together with the position and state of the read/write head. Other classical abstract computing devices, such as \emph{cellular automata} or \emph{counter machines}, can also be cast in this form. All these models share a key feature: they can be described as \emph{symbolic dynamical systems}, i.e.\ systems whose state is a sequence over a finite alphabet. A natural topological model for such systems is provided by the \emph{Cantor space}, which captures infinite symbolic sequences equipped with a natural metric. This space allows one to rigorously formalize the dynamics of symbolic computers and to analyze questions such as reachability and halting in a topological and operator-theoretic setting.  

In contrast, the field of analog computation investigates computational processes in continuous dynamical systems. Here the state space is typically a subset of $\mathbb{R}^n$ or a smooth manifold, and the evolution is governed by continuous (or even differentiable) maps or flows. Examples include continuous-time differential equations, hybrid systems, and physical analog computers. In this paradigm, computation is no longer carried out on discrete symbolic states, but rather on real-valued trajectories that evolve under deterministic dynamics. Analog models often highlight the tension between computability and the need for infinite precision in real number computations.  

We view both symbolic and analog computation as special instances of a dynamical systems, and use the Koopman operator framework as a unifying lens.

\subsection{The Cantor space}
The Cantor space over a finite alphabet $\Sigma$ is defined as the set of all infinite sequences of symbols,
\begin{equation}
\Sigma^{\mathbb{N}} = \{\, x = x_0 x_1 x_2 \ldots : x_k \in \Sigma \ \text{for all } k \in \mathbb{N}\,\}.
\end{equation}
Equipped with the metric $d : X \times X \to \mathbb{R}$,
\begin{eqnarray}
&&d(x,x') = 2^{-k}  \text{  if and only if } \nonumber \\
&&x_0 = x'_0, \, x_1 = x'_1, \ldots, x_{k-1} = x'_{k-1}, \ \text{but } x_k \neq x'_k,
\end{eqnarray}
the Cantor space becomes a compact metric space. Intuitively, two sequences are close whenever they share a long common prefix. For instance, a sequence of states $\mathbf{x}_\ell$ converges to a limit $\mathbf{x}_\infty$ (as $\ell \to \infty$) precisely when, for each $k$, the $k$-th symbol of $\mathbf{x}_\ell$ eventually stabilizes to the $k$-th symbol of $\mathbf{x}_\infty$. This metric is therefore often referred to as the \emph{topology of pointwise convergence}.  

In this topology, \emph{cylinders} play a fundamental role. A cylinder set is defined by a finite prefix,  
\begin{equation}
[x_0 x_1 \ldots x_k] = \{\, y \in \Sigma^{\mathbb{N}} : y_0 = x_0, \, y_1 = x_1, \ldots, y_k = x_k \,\}.
\end{equation}
Geometrically, these sets are both closed balls of radius $2^{-k}$ and open balls of radius $2^{-k+1}$; hence they are \emph{clopen}, i.e.\ simultaneously closed and open. Moreover, for each fixed radius there are only finitely many distinct cylinders.  

Every clopen set of the Cantor space can in fact be written as a finite union of cylinders. Open sets are arbitrary (finite or infinite) unions of such clopen sets, while closed sets are arbitrary intersections of them. This combinatorial structure makes the Cantor space particularly suitable for describing symbolic dynamical systems and their computational properties.

It is important to note that the choice of the alphabet $\Sigma$ is immaterial for the topology of the Cantor space. Indeed, for any two finite alphabets $\Sigma_1$ and $\Sigma_2$, the spaces $\Sigma_1^{\mathbb{N}}$ and $\Sigma_2^{\mathbb{N}}$ can be recoded into one another in a way that preserves the topological structure (open and closed sets), though not necessarily the specific metric. Similarly, spaces such as $\Sigma^{\mathbb{Z}}$ or $\Sigma^{\mathbb{Z}^2}$ can be trivially recoded into $\Sigma^{\mathbb{N}}$ by fixing an enumeration of $\mathbb{Z}$ or $\mathbb{Z}^2$ into $\mathbb{N}$.  

A map $F : \Sigma^{\mathbb{N}} \to \Sigma^{\mathbb{N}}$ is continuous precisely when each symbol of $F(\mathbf{x})$ depends only on finitely many symbols of $\mathbf{x}$. This condition captures the local update rules characteristic of symbolic dynamics.  

From this perspective, it becomes clear that Turing machines, where the global state encodes both the tape content and the head state, evolve through repeated application of a continuous map on the Cantor space. The same reasoning applies to cellular automata and, more generally, to all standard symbolic models of deterministic computation.  

For further details, see \cite{kurka2003topo}.

\section{Koopman operator and the reachability problem} \label{sec:koop-reach}
The halting problem is classically defined for Turing machines, but natural counterparts also exist for other models of computation such as cellular automata or counter machines. Over the years, several generalizations to arbitrary dynamical systems (both symbolic and continuous) have been proposed.  

One common formalization is the \emph{set-to-set reachability problem}: given two sets $A$ and $B$ in the state space, does there exist a trajectory that starts in $A$ and eventually passes through $B$? Importantly, this formulation does not require $B$ to be invariant (i.e.\ absorbing) under the dynamics, nor does it demand an exact initial condition. Instead, it suffices that the initial state lies somewhere in $A$. This choice is natural, since initializing a computer (symbolic or analog) with infinite precision at a single point is physically unrealistic. Moreover, $A$ and $B$ must admit a finite description, such as cylinders or clopen sets in the symbolic case \cite{asarin2000approx,delvenne2006computational}.  

Other formulations of the halting problem for dynamical systems, such as \emph{point-to-point reachability}, have also been studied. For classical models of computation (Turing machines, cellular automata, counter machines), all reasonable variants are equally undecidable (r.e.-complete). However, for broader classes of dynamical systems, different formulations can lead to distinct decidability properties \cite{boj2015reachability,asarin2000approx}.  
We stress that below in Section \ref{sec:koopman-automata-halting}, when considering finite automaton, we consider a different variant than in this section.

In this work, we introduce a \emph{Koopman-based abstraction} of the reachability problem. Specifically, we represent the dynamical system not directly on its state space but in the space of observables. Let $\mathcal{C}(X)$ denote the space of continuous complex-valued functions on $X$, which is a Banach space under the supremum norm. The dynamics is lifted to this space via the Koopman operator,
\begin{equation}
\mathcal{K}: \mathcal{C}(X) \to \mathcal{C}(X),
\end{equation}
defined by $(\mathcal{K}g)(x) = g(F(x))$ for $g \in \mathcal{C}(X)$.  

To connect this operator-theoretic setting with computability, we assume that $\mathcal{C}(X)$ is endowed with a \emph{computable structure}, in the sense of Pour-El and Richards \cite{pourel1989computability}. This allows one to meaningfully discuss the algorithmic properties of $\mathcal{K}$, and in particular to reformulate halting and reachability questions in terms of operator spectra and resolvents.

To make the connection between computation and dynamics precise, we restrict attention to a \emph{dense countable family of observables} $(g_k)_{k \in \mathbb{N}} \subset \mathcal{C}(X)$. These observables play the role of functions with a finite description, indexed by $k \in \mathbb{N}$, in analogy with the rationals (dense and countable) within the reals. We assume that:  
1. Each norm $\|g_k\|$ is a computable real number, i.e.\ it can be effectively approximated by rationals to arbitrary precision.  
2. The family $(g_k)$ is closed under finite linear combinations with rational complex coefficients.  
3. Since $\mathcal{C}(X)$ is a Banach algebra under pointwise multiplication, we further assume that $(g_k)$ is closed under product.  

These conditions guarantee that the family $(g_k)$ provides an effective algebra of observables dense in $\mathcal{C}(X)$.  

We now introduce the notion of a \emph{computable linear operator} in this setting.  

\begin{definition}[Computable operator]
Let $(g_k)_{k \in \mathbb{N}} \subset \mathcal{C}(X)$ be a family of observables as above.
A continuous linear operator $\mathcal{P} : \mathcal{C}(X) \to \mathcal{C}(X)$ is said to be \emph{computable} if the following holds. Suppose $(g_{k_n})_{n \in \mathbb{N}}$ is a Cauchy sequence converging rapidly to $g \in \mathcal{C}(X)$, i.e.\
\begin{equation}
\| g_{k_n} - g \| \leq 2^{-n}.
\end{equation}
Then, from the sequence of indices $(k_n)$ one can algorithmically compute a new sequence $(k_m)$ such that $(g_{k_m})_{m \in \mathbb{N}}$ is a fast Cauchy sequence converging to $\mathcal{P} g$. In other words, arbitrarily accurate approximations of $\mathcal{P} g$ can be effectively obtained from arbitrarily accurate approximations of $g$.
\end{definition}

In particular, we assume that the Koopman operator $\mathcal{K}$ associated with $\mathbf{F}$ is computable in this sense. This reflects the intuition that a single update step of the dynamical system should be computationally simple, while complexity arises from properties of arbitrarily long trajectories (as in the halting problem for Turing machines).  

To begin with, we assume that the Koopman operator $\mathcal{K}$ is computable. This assumption reflects the intuition that, in analogy with Turing machines or cellular automata, a single update step of the system should be algorithmically simple to describe and predict. The richness and potential difficulty of the dynamics only emerge from the properties of arbitrarily long trajectories, precisely as in the classical halting problem for Turing machines.  

Recall that the \emph{spectrum} of the Koopman operator is defined as
\begin{eqnarray}
&&\sigma(\mathcal{K}) \nonumber \\
&&= \{ \lambda \in \mathbb{C} \, : \, (\lambda I - \mathcal{K})^{-1}\nonumber \\
&&\hspace{1cm}\text{ does not exist as a bounded linear operator} \}.\nonumber
\end{eqnarray}
This set includes the eigenvalues, i.e.\ scalars $\lambda$ for which there exists $\varphi \neq 0$ with $\mathcal{K}\varphi = \lambda \varphi$.  

It is a standard result that $\mathcal{K}$ always admits a dominant eigenvalue $\lambda = 1$, corresponding to the constant observable $\phi(x) \equiv 1$. All other spectral values lie inside the unit disk, i.e.\ $|\lambda| \leq 1$. Moreover, whenever $|\lambda| > 1$, the resolvent operator can be expressed by the convergent Neumann series
\begin{equation}
(\lambda I - \mathcal{K})^{-1} = \lambda^{-1} \sum_{k=0}^{\infty} \lambda^{-k} \, \mathcal{K}^k,
\end{equation}
which converges in the operator norm.  

To connect these spectral properties with the halting problem, we focus on symbolic dynamical systems defined over the Cantor space. In this setting, a natural dense family of observables is given by the \emph{piecewise constant continuous functions} taking finitely many rational (real or complex) values. Their level sets are clopen sets, and hence these observables can be spanned by characteristic functions of clopen subsets of the Cantor space. In particular, such characteristic functions $\chi_A$ take values in $\{0,1\}$ and form a basis for describing finitely describable states of the system.  

Thus, in the symbolic case, the Koopman operator acts on a dense algebra of piecewise constant observables, and questions of reachability (and hence halting) can be rephrased in terms of the action of the resolvent $(\lambda I - \mathcal{K})^{-1}$ on these characteristic functions.
Consider first the symbolic case, and let $g_B = \chi_B$ be the characteristic function of a clopen set $B \subset X$. Then the action of the resolvent satisfies
\begin{equation}
(\lambda I - \mathcal{K})^{-1} g_B (x) \neq 0
\quad \Longleftrightarrow \quad
\exists \, t \in \mathbb{N} \ \text{such that } \mathbf{F}^t(x) \in B.
\end{equation}
In other words, $(\lambda I - \mathcal{K})^{-1} g_B$ takes nonzero values exactly on those points whose trajectories eventually enter $B$.  

Now let $g_A = \chi_A$ denote the characteristic function of another clopen set $A \subset X$. The observable
\begin{equation}
g_A \, (\lambda I - \mathcal{K})^{-1} g_B
\end{equation}
is identically zero if and only if no trajectory starting in $A$ ever reaches $B$. Hence, deciding whether
\begin{equation}
g_A \, (\lambda I - \mathcal{K})^{-1} g_B \neq 0
\end{equation}
is precisely equivalent to solving the reachability problem: does there exist an orbit beginning in $A$ that eventually intersects $B$?  

This observation suggests a general abstraction. If the problem is undecidable, the dynamical system exhibits nontrivial computational power, comparable to Turing universality. If the problem is decidable, the system must have strictly weaker, sub-Turing computational capabilities.  

\begin{definition}[Koopman Halting Problem]
Let $(g_k)_{k \in \mathbb{N}}$ be a dense, computable family of observables in $\mathcal{C}(X)$, and let $\mathcal{K}$ denote the Koopman operator associated with a dynamical system $F : X \to X$ on a compact metric space. For integers $a,b$ and rational $\lambda \in \mathbb{Q}$ with $|\lambda| > 1$, the \emph{Koopman Halting Problem} is the decision problem:
\begin{equation}
\text{Does } \quad g_a \, (\lambda I - \mathcal{K})^{-1} g_b \neq 0 ?
\end{equation}
Equivalently, the problem asks whether the set $B$ associated with $g_b$ is reachable from the set $A$ associated with $g_a$ under the dynamics of $F$. 
\end{definition}
This picture is shown in Fig.~\ref{fig:1}.
In what follows, we establish that the Koopman Halting Problem is \emph{recursively enumerable} (r.e.).  

\begin{theorem}
Let $\mathbf{F}: X \to X$ be a dynamical system on a compact metric space $X$, with associated Koopman operator $\mathcal{K}$ acting on complex-valued observables. Suppose there exists a countable family $(g_k)_{k \in \mathbb{N}} \subset \mathcal{C}(X)$ such that:
\begin{enumerate}
    \item linear combinations with rational coefficients are computable,
    \item pointwise products are computable,
    \item the action of $\mathcal{K}$ is computable,
\end{enumerate}
in the sense of Pour-El and Richards. Then the Koopman Halting Problem is recursively enumerable.  
\end{theorem}

\begin{proof}
Assume that $g_a (\lambda I - \mathcal{K})^{-1} g_b \neq 0$ for some $a,b \in \mathbb{N}$ and rational $\lambda$ with $|\lambda|>1$. Using the Neumann series, we have
\begin{equation}
g_a (\lambda I - \mathcal{K})^{-1} g_b
= \lambda \sum_{n=0}^{\infty} g_a (\lambda^{-1}\mathcal{K})^n g_b.
\end{equation}
The terms of this series decrease at a predictable rate $\mathcal{O}(|\lambda^{-1}|^n)$. Since the Koopman operator and algebraic operations are computable on $(g_k)$, each partial sum can be approximated effectively and to arbitrary precision.  

If $g_a (\lambda I - \mathcal{K})^{-1} g_b \neq 0$, then its norm is bounded below by some $2^{-\ell}$. A fast Cauchy sequence of approximations will eventually yield an estimate within $2^{-\ell-2}$ of the true value, from which we can algorithmically certify that the norm is at least $2^{-\ell-2} > 0$. Thus, whenever the answer is \emph{yes}, we can prove it in finite time. This establishes recursive enumerability.  
\end{proof}

\noindent
As in the classical case of Turing machines, the Koopman Halting Problem is therefore r.e. If it were shown to be r.e.-complete (that is, as hard as the halting problem of a universal Turing machine), this would imply that the underlying dynamical system is \emph{computationally universal}. Conversely, if the Koopman Halting Problem is decidable for a given system, then that system cannot match the computational power of a universal Turing machine.  

This formulation has the advantage of being both general and flexible. It does not depend on the specific nature of the state space (whether Euclidean, Cantor, or otherwise), but only on the choice of an effective dense family $(g_k)_{k \in \mathbb{N}}$. Examples include characteristic functions of clopen sets in Cantor space, polynomial observables, or piecewise linear maps on hypercubes $[0,1]^n$. Moreover, the construction extends naturally to continuous-time systems, and it is grounded in the well-established Pour-El–Richards computability theory for Banach spaces.

\section{Equicontinuous Symbolic Dynamical Systems}

In this section we show that equicontinuous symbolic systems admit a \emph{decidable} Koopman Halting Problem, and hence cannot serve as universal computers. This result is consistent with the findings of Delvenne–Kůrka–Blondel (DKB) on symbolic dynamical systems, although it is not strictly equivalent due to differences in the definition of universality.  

The theorem also highlights a key intuition: computational universality requires a certain degree of sensitivity to initial conditions, whereas equicontinuity guarantees the opposite property. Beyond its intrinsic value, this case study illustrates how the Koopman Halting Problem provides a rigorous framework to investigate the interplay between dynamical and computational properties.  

\begin{definition}[Equicontinuity]
A dynamical system $F : X \to X$ on a compact metric space $X$ is said to be \emph{equicontinuous} if for every $\epsilon > 0$, there exists $\delta > 0$ such that for all $x, x' \in X$ with $d(x,x') < \delta$, we have
\begin{equation}
d(F^n(x), F^n(x')) < \epsilon \quad \text{for all } n \in \mathbb{N}.
\end{equation}
Equivalently, the family $\{F^n : n \in \mathbb{N}\}$ is equicontinuous. Intuitively, nearby initial conditions remain close at all future times, and the system exhibits no sensitivity to initial conditions.
\end{definition}

\begin{theorem}
Let $\mathbf{F} : X \to X$ be an equicontinuous symbolic dynamical system on the Cantor space $X$. For any piecewise constant continuous observable $g \in \mathcal{C}(X)$, the span of its iterates
\begin{equation}
\operatorname{span}\{ g, \mathcal{K}g, \mathcal{K}^2 g, \mathcal{K}^3 g, \ldots \}
\end{equation}
is finite-dimensional, and its dimension is computable.
\end{theorem}

\begin{proof}
Cylinders of the form $A = [x_0 x_1 \ldots x_\ell]$ in the Cantor space are precisely the open balls of radius $2^{-\ell}$. For any fixed $\ell$, there are only finitely many such balls.  

By equicontinuity, the preimages $F^{-n}(A)$ form unions of $\delta$-balls for sufficiently small $\delta > 0$. Since there are finitely many $\delta$-balls, the sequence of preimages must eventually become periodic, i.e.\ $F^{-n}(A) = F^{-m}(A)$ for some $m > n \geq 0$.  

For $g = \chi_A$, the characteristic function of a cylinder $A$, this implies that the sequence
\begin{equation}
\mathcal{K}g, \ \mathcal{K}^2 g, \ \mathcal{K}^3 g, \ \ldots
\end{equation}
is eventually periodic, and hence generates a finite-dimensional subspace of $\mathcal{C}(X)$.  

Since every piecewise-constant observable $g$ is a finite linear combination of cylinder characteristic functions, it follows that the span of its iterates under $\mathcal{K}$ is finite-dimensional. Moreover, the dimension is computable, since the periodicity of preimages can be detected algorithmically.  
\end{proof}

\noindent
As a consequence, expressions of the form
\begin{equation}
g_a (\lambda I - \mathcal{K})^{-1} g_b
\end{equation}
can be reduced to explicit computations in a finite-dimensional subspace spanned by $g_a, g_b, \mathcal{K}g_b, \mathcal{K}^2 g_b, \ldots$  Equality with zero can then be decided algorithmically. Thus, for equicontinuous symbolic systems, the Koopman Halting Problem is decidable.



\begin{figure*}
    \centering
    \includegraphics[width=0.95\linewidth]{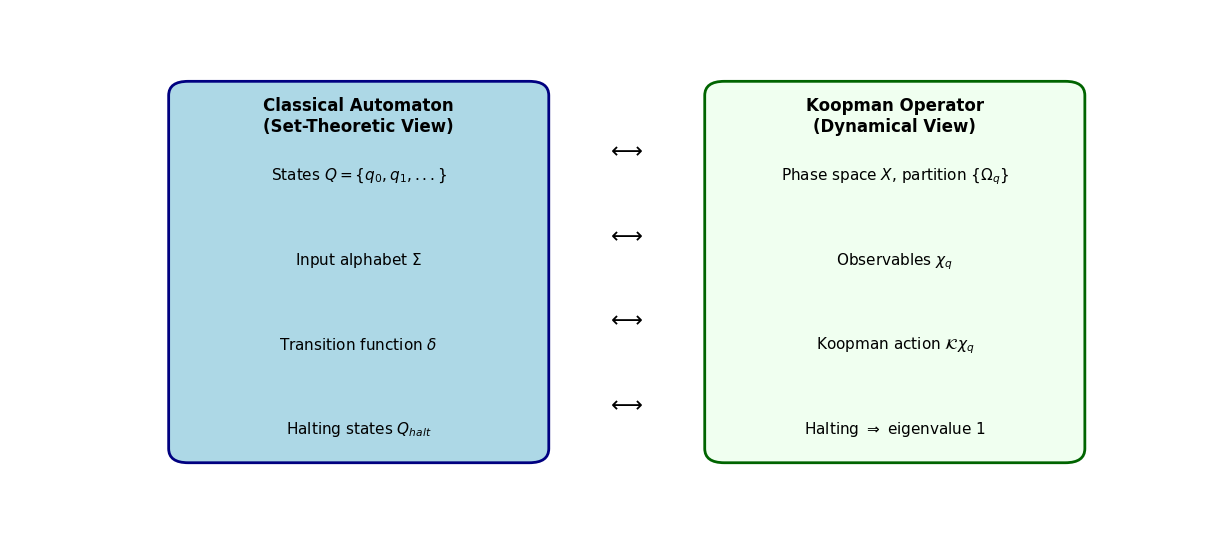}
    \caption{Mapping between classical automata (left) and Koopman operator framework (right). States correspond to partitions, transitions to Koopman action, and halting states to absorbing eigenfunctions.}
    \label{fig:2}
\end{figure*}
\section{Koopman Operators in the Non-Autonomous Case}

\subsection{External inputs}
So far, we have considered the Koopman operator for autonomous systems. In many models of computation, however, the dynamics are not strictly autonomous but rather \emph{driven by external inputs}. A canonical example is a finite automaton, where transitions depend not only on the current state but also on a symbol read from an external tape.  

In the discrete-time setting, such a non-autonomous system can be written as
\begin{equation}
  x_{k+1} \;=\; F(x_k, u_k),
  \label{eq:discrete-nonautonomous}
\end{equation}
where $x_k \in X$ is the system state and $u_k \in U$ is an external input (e.g.\ a tape symbol or control signal) at time $k$. To capture the joint evolution, we define an extended state
\begin{equation}
z_k = (x_k, u_k) \in Z = X \times U,
\end{equation}
with dynamics
\begin{equation}
z_{k+1} = G(z_k).
\end{equation}
On this extended space, one can define a Koopman operator $\mathcal{K}$ acting on observables $\phi : Z \to \mathbb{C}$ by
\begin{equation}
(\mathcal{K}\phi)(z) = \phi(G(z)).
\end{equation}
The operator remains linear in $\phi$, but the presence of the input $u_k$ makes the system no longer purely autonomous in $x_k$.  

\medskip

This viewpoint aligns naturally with models of digital computation. At a high level, a classical digital computer can be abstracted as a finite or pushdown automaton:  
\begin{itemize}
  \item A \emph{finite automaton} has a finite set of internal states $Q$, with transitions driven by input symbols.  
  \item A \emph{pushdown automaton} extends this model with a stack, allowing unbounded memory growth and shrinkage.  
\end{itemize}

From a dynamical perspective, one can imagine a physical substrate (e.g.\ transistors, voltages, signals) whose high-dimensional continuous dynamics collapse, under coarse-graining, to a finite set $Q$ of discrete machine states. The input sequence $(u_k)$ corresponds to symbols provided by the tape or the user, and determines the transitions between these coarse-grained regions of phase space.  

Formally, let $\{\Omega_s\}_{s \in Q}$ be a partition of the state space $X$ into regions labeled by machine states. We may then define observables that detect which region the system occupies:
\begin{equation}
  g_s(x) \;=\;
  \begin{cases}
    1, & \text{if } x \in \Omega_s, \\
    0, & \text{otherwise},
  \end{cases}
  \quad s \in Q.
\end{equation}
These $g_s$ are coarse-grained characteristic functions, and under the Koopman operator they evolve according to the transition logic of the underlying automaton.  
Under the extended Koopman operator $\mathcal{K}$ (which accounts for both the state $\mathbf{x}$ and the input $\mathbf{u}$), each coarse-grained observable $g_s$ evolves by composition:
\begin{equation}
  (\mathcal{K} g_s)(\mathbf{x}, \mathbf{u})
  \;=\;
  g_s\bigl(F(\mathbf{x},\mathbf{u})\bigr).
\end{equation}
At a discrete step, $(\mathcal{K} g_s)$ simply checks whether the system flows from region $\Omega_s$ to $\Omega_{s'}$ under the given input. In this way, the Koopman operator reproduces the transition logic of the automaton in a purely operator-theoretic framework.  

\medskip
\noindent\textbf{Example.}  
Consider a two-state automaton with $Q=\{q_0,q_1\}$ and input alphabet $\Sigma=\{a,b\}$. Its transition rules are:
\begin{equation}
q_0 \xrightarrow{a} q_1,
\quad
q_0 \xrightarrow{b} q_0,
\quad
q_1 \xrightarrow{a} q_1,
\quad
q_1 \xrightarrow{b} q_0.
\end{equation}
We define a partition of the state space $X$ into $\Omega_{q_0}$ and $\Omega_{q_1}$, and corresponding observables $g_{q_0}, g_{q_1}$. If the system is currently in $\Omega_{q_0}$ and receives input $a$, then
\begin{equation}
(\mathcal{K} g_{q_1})(x,a) = 1 \quad \text{and} \quad (\mathcal{K} g_{q_0})(x,a) = 0,
\end{equation}
meaning the system transitions into $\Omega_{q_1}$. Similarly, if the input is $b$, the Koopman operator preserves $g_{q_0}$, reflecting the self-loop at $q_0$. In this way, the Koopman operator acting on characteristic observables exactly recovers the automaton’s transition graph.  

\medskip
In the purely \emph{autonomous} case (such as a finite automaton with no external input), the system transitions depend only on the current state. In contrast, most realistic computational models include external inputs (e.g.\ tape symbols or user commands). This corresponds to the discrete-time formulation
\begin{equation}
  x_{k+1} = F(x_k, u_k),
\end{equation}
where $u_k$ is the input at step $k$. Defining the extended state $z_k=(x_k,u_k)$ yields an autonomous map
\begin{equation}
z_{k+1} = G(z_k),
\end{equation}
and hence a Koopman operator on the extended space. Observables are now functions of both $(x,u)$, and the operator $\mathcal{K}$ remains linear in this function space. This construction allows one to capture both autonomous and input-driven dynamics within the same operator-theoretic framework.  

\subsection{Koopman formalization of coarse-grained automata and halting}
\label{sec:koopman-automata-halting}

In the introduction, we framed computation in dynamical systems primarily in terms of 
\emph{reachability}: whether a trajectory starting in one region of state space can eventually arrive in another. 
This is the most general notion and applies both to symbolic and analog systems. 
However, reachability alone does not capture the full classical notion of ``halting'' familiar from Turing machines, 
where a computation terminates once the machine enters a designated absorbing configuration. 
In order to speak of halting explicitly, one must assume additional structure on the system. 
In particular, we now focus on \emph{driven} systems, i.e. dynamical systems with external inputs, 
together with a coarse-graining of the state space into finitely many regions corresponding to logical states. 
This coarse-graining allows us to identify absorbing subsets that play the role of halting states, 
and to characterize them spectrally via the Koopman operator.  

\medskip

We therefore view a digital computer (or more generally, any abstract automaton) 
as a coarse-grained driven dynamical system. 
Let $X$ denote a high-dimensional physical phase space, and suppose we partition it into finitely many 
\emph{coarse-grained regions} $\{\Omega_q\}_{q \in Q}$, where $Q$ is a finite index set of internal states. 
Each region $\Omega_q \subset X$ collects all microscopic states consistent with the logical state $q$ of the automaton.  

For each $q \in Q$, define the characteristic observable
\begin{equation}
  \chi_q(x) =
  \begin{cases}
    1, & x \in \Omega_q, \\
    0, & x \notin \Omega_q,
  \end{cases}
\end{equation}
which detects whether the system is in state $q$.  

In the driven case, the dynamics evolve according to
\begin{equation}
  x_{k+1} = F(x_k, u_k),
\end{equation}
with external input $u_k \in U$.  
To capture the joint state and input, we extend the phase space to $Z = X \times U$, where $z_k = (x_k, u_k)$. 
The Koopman operator acts on observables $\phi : Z \to \mathbb{C}$ via
\begin{equation}
  (\mathcal{K}\phi)(x,u) = \phi(F(x,u), u').
\end{equation}
This setting is illustrated schematically in Fig.~\ref{fig:2}.  

In the coarse-grained picture, a transition $q \to q'$ under input symbol $\sigma \in \Sigma$ means that 
whenever $x_k \in \Omega_q$ and $u_k$ corresponds to $\sigma$, the next state satisfies 
$x_{k+1} \in \Omega_{q'}$ (deterministically or with high probability). Formally,
\begin{equation}
  F(\Omega_q, \sigma) \subseteq \Omega_{q'}.
\end{equation}
This induces a transition graph $q \xrightarrow{\sigma} q'$, reproducing the structure of a finite automaton 
(or, with additional memory, a pushdown automaton).  

If the input space $U$ is continuous or high-dimensional, but only certain symbolic distinctions are relevant,  we introduce a coarse-graining map
\begin{equation}
  \Theta : U \;\to\; \Sigma,
\end{equation}
where $\Sigma$ is a finite alphabet and $U_\sigma = \Theta^{-1}(\sigma)$ denotes the region of input space 
corresponding to symbol $\sigma$. 
Each input $u_k$ therefore determines a unique symbol $\sigma = \Theta(u_k)$, 
which drives the transition $q \mapsto q'$ in the coarse-grained automaton.  

\medskip

\textit{Halting states.}  
In classical automata or Turing machines, one designates a subset $Q_{\mathrm{halt}} \subseteq Q$ of \emph{halting states}. 
Once the machine enters any $q_{\mathrm{halt}} \in Q_{\mathrm{halt}}$, it ceases further processing. 
This is in contrast with the reachability problem studied earlier, where the target set need not be absorbing. 
We now formalize halting in the dynamical setting:  

\begin{definition}[Halting state]
Let $F : X \to X$ be a discrete-time dynamical system on a compact metric space $X$, 
and suppose the state space is coarse-grained into finitely many disjoint regions 
$\{\Omega_q\}_{q \in Q}$. A set $\Omega_h \subseteq X$ is called a \emph{halting region}  
if it is absorbing under the dynamics:
\begin{equation}    
F(\Omega_h) \subseteq \Omega_h.
\end{equation}
Equivalently, once a trajectory enters $\Omega_h$, it remains there for all subsequent iterations.  

In terms of the Koopman operator $\mathcal{K}$, the characteristic observable $\chi_{\Omega_h}$ 
of the halting region satisfies
\begin{equation}    
\mathcal{K} \chi_{\Omega_h} = \chi_{\Omega_h},
\end{equation}
so $\chi_{\Omega_h}$ is an eigenfunction of $\mathcal{K}$ with eigenvalue $1$.  
Thus, halting in a coarse-grained driven dynamical system corresponds to reachability 
into an absorbing region, after which the system remains trapped there forever.
\end{definition}

\medskip

From the Koopman perspective, then, a halting condition corresponds precisely to an absorbing region in phase space: 
once the system’s coarse-grained state is $q_{\mathrm{halt}}$, the observable $\chi_{q_{\mathrm{halt}}}$ is invariant.  

In the classical theory of computation, a problem is \emph{decidable} if there exists a finite procedure that halts 
in a yes/no state for every input string in some language. In the Koopman viewpoint this becomes:  
\begin{itemize}
\item If an input string belongs to the language, then the system eventually flows into $\Omega_{q_{\mathrm{halt}}}$.  
\item If an input string is not in the language, the system never reaches $\Omega_{q_{\mathrm{halt}}}$.  
\end{itemize}

\begin{remark}
There is an important distinction between finite automata and Turing machines. 
In finite automata, acceptance depends on the sequence of external inputs (an open-loop process). 
By contrast, in Turing machines the system is autonomous once initialized: 
the computation unfolds solely from the initial condition. 
The Koopman formulation is compatible with both perspectives, 
but the notion of undecidability applies in the Turing case.  
\end{remark}

\medskip

\noindent\textbf{Example.}  
Let $Q=\{q_0,q_1\}$ be the set of states and $\Sigma=\{a,b\}$ the input alphabet, with transition rules
\begin{equation}
q_0 \xrightarrow{a} q_1, 
\quad q_0 \xrightarrow{b} q_0, 
\quad q_1 \xrightarrow{a} q_1, 
\quad q_1 \xrightarrow{b} q_1,
\end{equation}
and no halting states.  

To embed this in a dynamical framework, partition the phase space $X$ into $\Omega_{q_0}, \Omega_{q_1}$ 
corresponding to the states, and the input space $\mathcal{U}$ into $U_a, U_b$ corresponding to the symbols. 
The update rule $F$ then satisfies
\begin{eqnarray}
&&F(\Omega_{q_0}, U_a) \subseteq \Omega_{q_1}, \\
&&F(\Omega_{q_0}, U_b) \subseteq \Omega_{q_0}, \\
&&F(\Omega_{q_1}, U_a) \subseteq \Omega_{q_1}, \\
&&F(\Omega_{q_1}, U_b) \subseteq \Omega_{q_1}.\\
\end{eqnarray}

Defining characteristic observables $\chi_{q_0}, \chi_{q_1}$, one can track how $\mathcal{K}\chi_{q_0}$ or $\mathcal{K}\chi_{q_1}$ evolves under different inputs. 
The Koopman operator acting on these observables exactly reproduces the transition graph of the automaton.  

If we were to introduce a halting state $q_h$, the corresponding absorbing condition would be
\begin{equation}
F(\Omega_{q_h}, U_\sigma) \subseteq \Omega_{q_h}
\quad \text{for all } \sigma \in \Sigma,
\end{equation}
so that $\chi_{q_h}$ becomes an eigenfunction with eigenvalue $1$.  

Thus, coarse-graining the continuous phase space $X$ into regions $\{\Omega_q\}_{q \in Q}$ 
yields a finite automaton viewpoint on a physical system. 
Inputs $u_k$ are mapped into symbols in $\Sigma$, while halting states correspond to absorbing regions. 
The Koopman operator, acting on the characteristic observables $\chi_q$, reproduces the transition structure of the automaton. 
Classical notions of halting and decidability therefore remain valid in this setting, 
but are reformulated in an operator-theoretic framework that also unifies them with reachability.

\subsection{Koopman and Set-Theoretic Automata}
\label{sec:koopman-vs-set-theoretic}

At first sight, the Koopman operator framework developed in Sections~\ref{sec:koopman-intro}--\ref{sec:koopman-automata-halting} may appear distant from the classical \emph{set-theoretic} approach in theoretical computer science. In the latter, one typically defines:
\begin{itemize}
\item a (finite or countable) set of states $Q$,  
\item a finite input alphabet $\Sigma$,  
\item a transition function, which depends on the model:
  \begin{itemize}
  \item For a \emph{finite automaton}: $\delta : Q \times \Sigma \to Q$,  
  \item For a \emph{pushdown automaton}: $\delta : Q \times \Sigma \times \Gamma \to Q \times \Gamma^*$, where $\Gamma$ is the stack alphabet,  
  \item For a \emph{Turing machine}: $\delta : Q \times \Sigma \to Q \times \Sigma \times \{L,R\}$,  
  \end{itemize}
\item and a subset of halting (or accepting) states $Q_{\mathrm{halt}} \subseteq Q$.  
\end{itemize}

In all these cases, the system is purely discrete: each configuration belongs to $Q$ (or an extended configuration space), and the transition function $\delta$ defines a labeled directed graph. Halting states are those in which the computation terminates.  

From a dynamical systems perspective, however, each logical state $q \in Q$ may correspond to a region of a high-dimensional continuous system $x \in X \subseteq \mathbb{R}^n$ (e.g.\ voltages in a circuit, atomic configurations, or neural activity patterns). One can define a partition $\{\Omega_q\}_{q \in Q}$ of $X$, where each $\Omega_q$ collects all microstates consistent with the logical state $q$.  

For each $q \in Q$, define the characteristic observable
\begin{equation}
\chi_q(x) =
\begin{cases}
1, & x \in \Omega_q, \\
0, & x \notin \Omega_q.
\end{cases}
\end{equation}
The Koopman operator $\mathcal{K}$ acts linearly on these functions:
\begin{equation}
(\mathcal{K}\chi_q)(x,u) = \chi_q(F(x,u)).
\end{equation}
In set-theoretic language: if $x \in \Omega_q$ and $u \in U_\sigma$, then $F(x,u) \in \Omega_{q'}$, corresponding to the transition $\delta(q,\sigma) = q'$.  

To incorporate inputs, partition the input space $U$ into $\{U_\sigma\}_{\sigma \in \Sigma}$ via a coarse-graining map $\Theta : U \to \Sigma$. Then the region $\Omega_q \times U_\sigma$ represents being in state $q$ while reading symbol $\sigma$. The next state $q'$ is determined by the set $\Omega_{q'}$ containing $F(\Omega_q, U_\sigma)$.  

Finally, a halting state $q_{\mathrm{halt}}$ corresponds to an absorbing region $\Omega_{q_{\mathrm{halt}}}$:
\begin{equation}
F(\Omega_{q_{\mathrm{halt}}}, U_\sigma) \subseteq \Omega_{q_{\mathrm{halt}}}
\quad \text{for all } \sigma \in \Sigma.
\end{equation}
In this case, $\chi_{q_{\mathrm{halt}}}$ is an eigenfunction of $\mathcal{K}$ with eigenvalue $1$. In the set-theoretic automaton model, this is expressed as $\delta(q_{\mathrm{halt}},\sigma) = q_{\mathrm{halt}}$ for all $\sigma$.  

\medskip
\noindent
Thus, the correspondence between the two viewpoints is direct:
\begin{itemize}
\item Discrete states $Q$ $\;\leftrightarrow\;$ partitions $\{\Omega_q\}$ of the continuous phase space $X$,  
\item Transition function $\delta$ $\;\leftrightarrow\;$ inclusions $F(\Omega_q, U_\sigma) \subseteq \Omega_{q'}$ and the Koopman action $\chi_q \mapsto \chi_{q'}$,  
\item Halting states $Q_{\mathrm{halt}}$ $\;\leftrightarrow\;$ absorbing sets $\Omega_{q_{\mathrm{halt}}}$ with $\chi_{q_{\mathrm{halt}}}$ invariant under $\mathcal{K}$.  
\end{itemize}
Thus, the classical automaton or Turing machine can be seen as a \emph{discrete quotient} of a more general continuous dynamical system. The Koopman operator lifts the transition graph into an operator-theoretic representation acting on characteristic observables.

\subsection{Eigenvalues of the Koopman Operator in the Presence of Halting States}
\label{sec:koopman-halting-eigenvalues}

When a dynamical system possesses \emph{halting} or \emph{absorbing} states, this has direct consequences for the spectral properties of its Koopman operator. Recall that in the coarse-grained automaton picture (Section~\ref{sec:koopman-automata-halting}), a halting state $q_{\mathrm{halt}}$ corresponds to a region $\Omega_{q_{\mathrm{halt}}} \subseteq X$ such that, once entered, the system never leaves.  

Formally, if the evolution map $F$ satisfies
\begin{equation}
  F(\Omega_{q_{\mathrm{halt}}}, u) \;\subseteq\; \Omega_{q_{\mathrm{halt}}},
  \quad
  \forall u \in U,
\end{equation}
then $\Omega_{q_{\mathrm{halt}}}$ is an \emph{absorbing} subset of the phase space $X$.  

Let $\chi_{q_{\mathrm{halt}}}$ be the characteristic observable of $\Omega_{q_{\mathrm{halt}}}$. Then the Koopman operator acts as
\begin{equation}
(\mathcal{K}\chi_{q_{\mathrm{halt}}})(x,u) = \chi_{q_{\mathrm{halt}}}(F(x,u)).
\end{equation}
If $x \in \Omega_{q_{\mathrm{halt}}}$, then $F(x,u) \in \Omega_{q_{\mathrm{halt}}}$ for all $u$, so
\begin{equation}
(\mathcal{K}\chi_{q_{\mathrm{halt}}})(x,u) = \chi_{q_{\mathrm{halt}}}(x) = 1.
\end{equation}
Thus, restricted to $\Omega_{q_{\mathrm{halt}}}$, $\chi_{q_{\mathrm{halt}}}$ is an eigenfunction of $\mathcal{K}$ with eigenvalue $\lambda = 1$.  

\medskip
We can strengthen this observation by considering the entire \emph{basin of attraction} of $\Omega_{q_{\mathrm{halt}}}$, defined as
\begin{equation}
B(\Omega_{q_{\mathrm{halt}}}) = \{\, x \in X : \exists \, n \in \mathbb{N} \ \text{with} \ F^n(x,u) \in \Omega_{q_{\mathrm{halt}}} \,\}.
\end{equation}
Then the characteristic function $\chi_{B(\Omega_{q_{\mathrm{halt}}})}$ is also invariant under $\mathcal{K}$, and hence is an eigenfunction with eigenvalue $1$. In fact, it is this basin observable—not just $\chi_{\Omega_{q_{\mathrm{halt}}}}$—that captures the long-term halting behavior of the system.  

\begin{theorem}
Let $X$ be a compact phase space, and let $H \subseteq X$ be an absorbing set under a discrete-time evolution law $F$ (i.e.\ $F(H) \subseteq H$). Then the characteristic function $\chi_{B(H)}$ of the basin of $H$ is an eigenfunction of the associated Koopman operator $\mathcal{K}$ with eigenvalue $1$.
\end{theorem}

\begin{proof}
By definition of the basin, if $x \in B(H)$ then there exists $n \geq 0$ such that $F^n(x) \in H$. Since $H$ is absorbing, $F^m(x) \in H$ for all $m \geq n$. Therefore, once $x$ enters the basin, all forward iterates remain in $B(H)$.  

Let $\chi_{B(H)}$ denote the characteristic function of $B(H)$. For any $x \in B(H)$, we have $F(x) \in B(H)$, so
\begin{equation}
(\mathcal{K}\chi_{B(H)})(x) = \chi_{B(H)}(F(x)) = 1 = \chi_{B(H)}(x).
\end{equation}
Thus $\chi_{B(H)}$ is an eigenfunction with eigenvalue $1$. Outside the basin, $\chi_{B(H)}(x) = 0$, and invariance need not hold pointwise, but on the invariant subspace spanned by $\chi_{B(H)}$, the eigenvalue relation is satisfied. 
\end{proof}

Thus, whenever an absorbing (halting) set exists, the Koopman operator necessarily has $\lambda = 1$ in its spectrum, since the system’s flow (or map) does not leave that set. Concretely, the subspace of observables supported on $\Omega_{q_{\mathrm{halt}}}$ (or more generally, on its basin of attraction) remains invariant under $\mathcal{K}$.  

\begin{corollary}
The multiplicity of the eigenvalue $1$ for the Koopman operator $\mathcal{K}$ is at least the number of distinct halting states.  
\end{corollary}

\begin{proof}
Each halting state $q_{\mathrm{halt}}$ has an associated basin of attraction $B(\Omega_{q_{\mathrm{halt}}})$. The characteristic function $\chi_{B(\Omega_{q_{\mathrm{halt}}})}$ is an eigenfunction of $\mathcal{K}$ with eigenvalue $1$. Distinct halting states have disjoint basins, yielding linearly independent eigenfunctions. In addition, even if there are no halting states, the constant function $1$ is always an eigenfunction with eigenvalue $1$. Hence the multiplicity of $\lambda = 1$ is at least equal to the number of halting states, and possibly greater. $\square$
\end{proof}

In many Markov-type or piecewise-deterministic systems with absorbing states, the remainder of the spectrum lies strictly inside the unit circle ($|\lambda| < 1$), apart from the $\lambda = 1$ eigenvalues associated with absorbing basins. Intuitively, all trajectories either converge to, or eventually enter, a halting basin. Repeated application of $\mathcal{K}$ then progressively suppresses the contribution of non-absorbing states, so that in the long-time limit the dynamics projects onto the subspace spanned by the eigenfunctions associated with halting states.  
It is important to note that the mere presence of $\lambda = 1$ in the Koopman spectrum does not by itself imply that the system will reach a halting state from every initial condition. In fact, $\lambda = 1$ is always present, since the constant function $\mathbf{1}(x) \equiv 1$ is trivially an eigenfunction with eigenvalue $1$. What the existence of an additional eigenfunction such as $\chi_{q_{\mathrm{halt}}}$ signifies is only that $\Omega_{q_{\mathrm{halt}}}$ is an absorbing region for \emph{some} subset of initial conditions.  

In automaton language, this corresponds to the fact that not all strings are necessarily accepted: but if an input string \emph{is} accepted, then the computation ends in the halting state $q_{\mathrm{halt}}$. From the operator perspective, the subspace spanned by $\chi_{q_{\mathrm{halt}}}$ is invariant with eigenvalue $1$, while the rest of the observable space may correspond to transients or recurrent behaviors, associated with eigenvalues of smaller magnitude or lying on the unit circle.  

\medskip
\noindent\textbf{Example.}  
Consider a discrete map $x_{k+1} = F(x_k)$ with two coarse-grained regions: $\Omega_1$ (non-halting) and $\Omega_{\mathrm{halt}}$ (absorbing). Then
\begin{equation}
\chi_{\mathrm{halt}}(F(x)) =
\begin{cases}
0, & x \in \Omega_1, \\
1, & x \in \Omega_{\mathrm{halt}},
\end{cases}
\end{equation}
assuming that once in $\Omega_{\mathrm{halt}}$, the system remains there. This shows that $\chi_{\mathrm{halt}}$ is an eigenfunction with eigenvalue $1$ (at least on $\Omega_{\mathrm{halt}}$). In contrast, $\chi_1$ need not be an eigenfunction: if $\Omega_1$ eventually maps into $\Omega_{\mathrm{halt}}$, then repeated application of $\mathcal{K}$ will progressively transfer weight from $\chi_1$ to $\chi_{\mathrm{halt}}$, until in the long-time limit $\chi_1$ is annihilated. This reflects the fact that the system halts with probability $1$ if all trajectories eventually reach $\Omega_{\mathrm{halt}}$.  

\medskip
Thus, in systems (or automata) with halting states, the Koopman spectrum has a clear structure:  
\textit{Absorbing modes:} eigenfunctions at $\lambda = 1$ corresponding to halting basins.  
 \textit{Transient modes:} eigenvalues with $|\lambda| < 1$, associated with states that eventually decay into halting regions.  
 \textit{Neutral/recurrent modes:} eigenvalues on the unit circle ($|\lambda| = 1$), corresponding to cycles or persistent dynamics.  

This division  (absorbing modes at $\lambda=1$, transient modes inside the unit disk and neutral modes on the boundary) is the hallmark of halting or absorbing phenomena in the Koopman framework.

\section{Dynamical Distance, Topology, and Graph-Theoretic Representation}

The Koopman framework reformulates computation in terms of observables and operator spectra. At the same time, classical automata theory relies on graph-theoretic and set-theoretic notions of state transitions. To bridge these perspectives, it is useful to endow the coarse-grained state space $Q$ with a natural \emph{dynamical distance} and corresponding topology. This construction shows how the transition graph of an automaton can be recovered as a topological and metric structure, ensuring that the Koopman operator viewpoint is consistent with the discrete set-theoretic one.

\medskip

Let $Q$ denote the set of coarse-grained states corresponding to a partition $\{\Omega_q\}_{q \in Q}$ of the phase space $X$. Define a transition relation on $Q$ by declaring that
\begin{equation}
q \leadsto q'
\quad \Longleftrightarrow \quad
\exists \, \sigma \in \Sigma \text{ such that } F(\Omega_q, U_\sigma) \subseteq \Omega_{q'}.
\end{equation}
This relation induces a directed graph $G = (Q, E)$, where $E$ consists of all pairs $(q,q')$ with $q \leadsto q'$. The graph $G$ represents the coarse-grained dynamics of the system, with nodes corresponding to states and edges encoding possible one-step transitions.

\medskip

We now define a \emph{dynamical distance} $d : Q \times Q \to [0,\infty]$ by
\begin{eqnarray}
d(q,q') =\begin{cases}
L(q,q^\prime), & \text{if such a path exists}, \\
\infty, & \text{otherwise},
\end{cases}\nonumber \\
\end{eqnarray}
where, $L(q,q^\prime)$ is defined as the shortest directed path from $q$ to $q'$.
This metric satisfies $d(q,q)=0$ and the triangle inequality whenever the distances are finite. In general, however, $d$ is not symmetric: a path from $q$ to $q'$ does not imply a path from $q'$ to $q$.

The distance $d$ induces a topology on $Q$, where the open neighborhoods are given by the balls
\begin{equation}
B(q,n) = \{ q' \in Q : d(q,q') < n \}, \qquad n \geq 1.
\end{equation}
Each such ball consists of all states reachable from $q$ in fewer than $n$ transitions along $G$.

Independently, one can define a \emph{dynamical topology} $\mathcal{T}$ on $Q$ by declaring a subset $O \subseteq Q$ to be open if, whenever $q \in O$, all immediate successors of $q$ are also contained in $O$. That is, $(q,q') \in E \implies q' \in O$ whenever $q \in O$. In other words, open sets are those invariant under one-step transitions of the graph.

\medskip

We now establish that these two approaches yield the same topology.

\begin{proposition}
The topology generated by the open balls $B(q,n)$ coincides with the dynamical topology $\mathcal{T}$ defined by the transition structure of $G$.
\end{proposition}

\begin{proof}
Suppose $O$ is open in $\mathcal{T}$. Then for each $q \in O$, all immediate successors $q'$ of $q$ are also in $O$. This implies $B(q,2) \subseteq O$. Since every $q \in O$ is contained in some $B(q,2) \subseteq O$, it follows that $O$ is open in the distance-induced topology.

Conversely, suppose $O$ is open in the topology induced by $d$. Then for each $q \in O$, there exists $n$ such that $B(q,n) \subseteq O$. In particular, $B(q,2) \subseteq O$, but $B(q,2)$ contains $q$ and all its immediate successors. Hence all immediate successors of $q$ belong to $O$, so $O$ is open in $\mathcal{T}$. This proves the equivalence.
\end{proof}

\medskip

This equivalence shows that the transition graph of a coarse-grained dynamical system induces both a graph-theoretic and a topological structure on $Q$, and that these two viewpoints are consistent. It provides a natural bridge between the discrete formalism of automata theory and the operator-theoretic Koopman perspective.
The directed graph $G$ defined by the transition relation is, by construction, equivalent to the transition graph of a deterministic finite automaton (DFA): the nodes represent internal states and the directed edges encode the transition function under the action of input symbols. In the case of a Turing machine, the graph $G$ captures the transition structure of the machine’s control automaton, encompassing both the internal state and the tape configuration.  

Within this framework, the dynamical distance $d$ recovers the minimal number of transitions required to move from one configuration (or state) to another, while the induced topology encodes the structure of reachable sets of configurations. In other words, the coarse-grained dynamical system, together with its graph, distance, and topology, provides a mathematical structure that is fully equivalent to the transition graph of a DFA or the control mechanism of a Turing machine.  

This perspective also enables us to connect the combinatorial structure of the transition graph with the spectral properties of the Koopman operator. In particular, cycles in the graph correspond to eigenvalue relations in $\mathcal{K}$. We now turn to this connection by analyzing the cyclic structure of $G$ and its implications for the Koopman spectrum.

We now show that cycles in the transition graph $G$ impose algebraic constraints on the Koopman operator. Let
\begin{equation}
C = (q_0, q_1, \dots, q_{m-1}, q_0)
\end{equation}
be a directed cycle in $G$, where $(q_j, q_{j+1}) \in E$ for $j=0,\dots,m-1$ (indices modulo $m$).  

For each edge $(q_{j-1},q_j)$ in the cycle, let $\mathcal{K}_j$ denote the Koopman operator restricted to that transition. On the characteristic observable $\chi_{q_{j-1}}$, we may write
\begin{equation}
\mathcal{K}_j \chi_{q_{j-1}} = \lambda_j \chi_{q_j},
\end{equation}
where $\lambda_j$ is the eigenvalue associated with this transition. In the simplest deterministic case, $\lambda_j=1$; more generally, $\lambda_j$ may encode weighting factors or stochastic structure.  

Consider now the composition of Koopman operators along the cycle:
\begin{equation}
\mathcal{K}_C = \mathcal{K}_m \mathcal{K}_{m-1} \cdots \mathcal{K}_1.
\end{equation}
Applying $\mathcal{K}_C$ to $\chi_{q_0}$ gives
\begin{equation}
\mathcal{K}_C \chi_{q_0} = \left( \prod_{j=1}^m \lambda_j \right) \chi_{q_0}.
\end{equation}
But since traversing the cycle returns the system to its initial macrostate $q_0$, we must also have
\begin{equation}
\mathcal{K}_C \chi_{q_0} = \chi_{q_0}.
\end{equation}
Hence,
\begin{equation}
\prod_{j=1}^m \lambda_j = 1.
\end{equation}

\begin{proposition}
For any directed cycle $C$ in the transition graph $G$, the product of the Koopman eigenvalues associated with the edges of $C$ must equal $1$.  
\end{proposition}

\begin{proof}
The composition of Koopman operators along the cycle maps $\chi_{q_0}$ to itself up to the factor $\prod_{j=1}^m \lambda_j$. Since the cycle returns the system to $q_0$, this action must be the identity, and thus the product of eigenvalues equals $1$.
\end{proof}

\medskip
This result shows that the cyclic structure of the transition graph constrains the spectrum of the Koopman operator. Each closed loop in $G$ enforces a multiplicative relation among eigenvalues along the cycle, reflecting the consistency of returning to the initial state after a finite sequence of transitions.

\section{Conclusion}

In this paper we have examined the Koopman operator as a framework for understanding computation in dynamical systems. Our central contribution has been the formulation of the \emph{Koopman Halting Problem}, which recasts classical notions of halting and reachability in terms of the resolvent of the Koopman operator. This perspective makes it possible to interpret computational properties in spectral and operator-theoretic terms rather than purely combinatorial or symbolic ones. In particular, we established that the Koopman Halting Problem is recursively enumerable in general computable dynamical systems, and that in the specific case of equicontinuous symbolic systems it is decidable, thereby ruling out universality. Halting states themselves {(as absorbing states)} appear in the Koopman spectrum as eigenfunctions with eigenvalue one, while the multiplicity of this eigenvalue reflects the number of distinct absorbing basins. Cycles in the transition graph impose algebraic constraints on the operator spectrum, ensuring consistency between graph-theoretic dynamics and operator evolution.  

The significance of this approach lies in its ability to bring together previously separate traditions in the study of dynamical computation. On the one hand, the classical results of Turing, Church, and their successors established computation as a dynamical process of state evolution in symbolic machines, while later work by Delvenne, Kůrka, and Blondel extended these ideas to general symbolic dynamical systems. On the other hand, Pour-El and Richards provided a rigorous theory of computability for Banach spaces, and applied mathematicians have developed Koopmanism as a powerful spectral method for nonlinear systems, particularly in control theory and data-driven modeling. By placing computational universality, decidability, and halting within the Koopman framework, we show that these strands can be reconciled within a single operator-theoretic language.  

This synthesis opens several perspectives. Computation can be understood spectrally: absorbing states correspond to invariant subspaces, transient dynamics to eigenvalues inside the unit disk, and cycles to unit-modulus eigenvalues. This decomposition parallels well-known structures in Markov chains and control systems, suggesting that the boundaries between computation, probability, and dynamical analysis are more porous than they might appear. It also points to a continuity between symbolic and analog computation: while the former is encoded in partitions of Cantor space and clopen observables, the latter can be expressed through polynomial or piecewise linear observables on Euclidean domains, with the Koopman operator unifying both settings. One may even speculate that quantum computation, whose dynamics is already governed by unitary operators, could be understood as another instance of the same operator-theoretic paradigm.  

Our results are modest steps in this direction, but they underline a general principle: the Koopman operator provides a rigorous and flexible framework in which computational properties of dynamical systems can be expressed, studied, and compared. By doing so, it builds a bridge between the operator-theoretic methods that dominate modern dynamical systems theory and the algorithmic questions that lie at the heart of theoretical computer science. We hope that this work will stimulate further interaction between these communities, and that the language of Koopmanism will become part of the vocabulary through which computation in dynamical systems is understood.

\begin{acknowledgments}
Both authors thank the organizers of the CSH Workshop ``Computation in Dynamical Systems", October 2023, Obergurgl (Austria), and the organizers of the SFI Workshops on stochastic thermodynamics and computation, where these ideas were developed.
FC is currently an employee of Planckian srl.
J-C D acknowledges funding from the Project INTER/FNRS/20/15074473 ‘TheCirco’ on Thermodynamics of Circuits for Computation, funded by the F.R.S.-FNRS (Belgium) and FNR (Luxembourg), and from the
SIDDARTA Concerted Research Action (ARC) of the Fédération Wallonie-Bruxelles (Belgium). J-C D is also a JSPS fellow and Invited Professor at University of Kyoto.
\end{acknowledgments}


\begin{thebibliography}{99}

\bibitem{church1936unsolvable}
Church, A. (1936). An unsolvable problem of elementary number theory. \emph{American Journal of Mathematics}, 58 (2), 345–363.

\bibitem{turing1936computable}
Turing, A. M. (1936). On computable numbers, with an application to the Entscheidungsproblem. \emph{Proceedings of the London Mathematical Society}, 42 (2), 230–265.


\bibitem{moore1990unpredictability}
Moore, C. (1990). Unpredictability and undecidability in dynamical systems. \emph{Physical Review Letters}, 64(20), 2354–2357.


\bibitem{hopcroft1979automata}
Hopcroft, J. E., \& Ullman, J. D. (1979). Introduction to Automata Theory, Languages, and Computation. Reading, MA: Addison-Wesley.



\bibitem{shannon1941}
Shannon, C. E. (1941). Mathematical theory of the differential analyzer. \emph{Journal of Mathematics and Physics}, 20(1–4), 337–354.



\bibitem{bournez2007survey}
Bournez, O., \& Campagnolo, M. L. (2007). A survey on continuous time computations. In B. Cooper, B. Löwe, \& A. Sorbi (Eds.), \emph{New Computational Paradigms} (pp. 383–423). Berlin/Heidelberg: Springer.

\bibitem{rubel1981universal}
Rubel, L. A. (1981). A universal differential equation. \emph{Bulletin of the American Mathematical Society}, 4(3), 345–349.



\bibitem{branicky1995universal}
Branicky, M. S. (1995). Universal computation and other capabilities of hybrid and continuous dynamical systems. \emph{Theoretical Computer Science}, 138(1), 67–100.



\bibitem{wolfram1984universality}
Wolfram, S. (1984). Universality and complexity in cellular automata. \emph{Physica D: Nonlinear Phenomena}, 10 (1–2), 1–35.


\bibitem{delvenne2006computational}
Delvenne, J.-C., Kurka, P., \& Blondel, V. D. (2006). Computational universality in symbolic dynamical systems. \emph{Natural Computing}, 5(4), 379–404.


\bibitem{blum1989theory}
Blum, L., Shub, M., \& Smale, S. (1989). On a theory of computation and complexity over the real numbers: NP-completeness, recursive functions and universal machines. \emph{Bulletin of the American Mathematical Society}, 21(1), 1–46.


\bibitem{mezic2005spectral}
Mezić, I. (2005). Spectral properties of dynamical systems, model reduction and decompositions. \emph{Nonlinear Dynamics}, 41(1), 309–325.


\bibitem{budisic2012applied}
Budišić, M., Mohr, R., \& Mezić, I. (2012). Applied Koopmanism. \emph{Chaos: An Interdisciplinary Journal of Nonlinear Science}, 22 (4), 047510.


\bibitem{ercsey2011optimization}
Ercsey-Ravasz, M., \& Toroczkai, Z. (2011). Optimization hardness as transient chaos in an analog approach to constraint satisfaction. \emph{Nature Physics}, 7 (12), 966–970.



\bibitem{pourel1974}
Pour-El, M. B. (1974). Abstract computability and its relation to the general purpose analog computer (some connections between logic, differential equations, and analog computers). \emph{Transactions of the American Mathematical Society}, 199, 1–28.


\bibitem{vergis1986}
Vergis, A., Steiglitz, K., \& Dickinson, B. (1986). The complexity of analog computation. \emph{Mathematics and Computers in Simulation}, 28 (2), 91–113.


\bibitem{siegelmann1995}
Siegelmann, H. T., \& Sontag, E. D. (1995). On the computational power of neural nets. \emph{Journal of Computer and System Sciences}, 50(1), 132–150.

\bibitem{moore1996}
Moore, C. (1996). Recursion theory on the reals and continuous-time computation. \emph{Theoretical Computer Science}, 162 (1), 23–44.

\bibitem{graca2003}
Graça, D. S., \& Costa, J. F. (2003). Analog computers and recursive functions over the reals. \emph{Journal of Complexity}, 19(5), 644–664.

\bibitem{bournez2006}
Bournez, O., Graça, D. S., \& Pouly, A. (2006). Computational analysis of continuous time dynamical systems. \emph{Electronic Notes in Theoretical Computer Science}, 155, 5–18.

\bibitem{bournez2017}
Bournez, O., Graça, D. S., \& Pouly, A. (2017). Polynomial time corresponds to solutions of polynomial ordinary differential equations of polynomial length. \emph{Journal of the ACM}, 64(6), 1–76.


\bibitem{krzakala2007gibbs}
Krzakala, F., Montanari, A., Ricci-Tersenghi, F., Semerjian, G., \& Zdeborová, L. (2007). Gibbs states and the set of solutions of random constraint satisfaction problems. \emph{Proceedings of the National Academy of Sciences}, 104 (25), 10318–10323.


\bibitem{papadimitriou1994complexity}
Papadimitriou, C. H. (1994). \emph{Computational Complexity}. Reading, MA: Addison-Wesley.



\bibitem{pourel1981wave}
Pour-El, M. B., \& Richards, J. I. (1981). The wave equation with computable initial data such that its unique solution is not computable. \emph{Advances in Mathematics}, 39 (3), 215–239.


\bibitem{koopman1932dynamical}
Koopman, B. O., \& von Neumann, J. (1932). Dynamical systems of continuous spectra. \emph{Proceedings of the National Academy of Sciences}, 18 (3), 255–263.

\bibitem{koopman1931hamiltonian}
Koopman, B. O. (1931). Hamiltonian systems and transformation in Hilbert space. \emph{Proceedings of the National Academy of Sciences}, 17(5), 315–318.


\bibitem{mauroy2020koopman}
Mauroy, A., Susuki, Y., \& Mezić, I. (2020). \emph{Koopman Operator in Systems and Control}(Vol. 7). Berlin/Heidelberg: Springer.
\bibitem{caravellilin}
Caravelli, F., \& Lin, Y.-T. (2025). The combinatorial structure of Lotka–Volterra equations and the Koopman operator. \emph{Physica A}, 670, 130484.


\bibitem{kurka2003topo}
Kurka, P. (2003). Topological and Symbolic Dynamics, \emph{ Paris: Société Mathématique de France}.


\bibitem{asarin2000approx}
Asarin, E., Bournez, O., Dang, T., \& Maler, O. (2000, March). Approximate reachability analysis of piecewise-linear dynamical systems. In \emph{International Workshop on Hybrid Systems: Computation and Control} (pp. 20–31). Berlin, Heidelberg: Springer.

\bibitem{boj2015reachability}
Bojanczyk, M., Lasota, S., \& Potapov, I. (2015). Reachability problems, in \emph{Lecture Notes in Computer Science}.

\bibitem{pourel1989computability}
Pour-El, M. B., \& Richards, J. I. (1989). Computability structures on a Banach space. In \emph{Computability in Analysis and Physics} (Vol. 1, pp. 77–93). Association for Symbolic Logic.








\end{thebibliography}
\end{document}